\documentclass[12p]{article}
\usepackage{graphicx}
\usepackage[a4paper, margin=1.1in]{geometry}
\usepackage{authblk}

\usepackage[numbers]{natbib}

\usepackage{amssymb}
\usepackage{amsmath}
\usepackage{amsthm}
\usepackage{amsfonts}
\usepackage{listings}
\usepackage{hyperref}
\usepackage{graphicx}
\usepackage{wrapfig}
\usepackage{xcolor}
\usepackage{mathtools}
\usepackage{algorithm}
\usepackage{algpseudocode}
\usepackage{textcomp}
\usepackage{makecell}

\newtheorem{theorem}{Theorem}[section]
\newtheorem{definition}{Definition}[section]

\theoremstyle{remark}
\newtheorem*{remark}{\textbf{Remark}}

\definecolor{pink}{rgb}{0.8, 0.2, 0.5}

\newcommand{\Hdim}{\calH_{\mathrm{dim}}}

\newcommand{\fiber}{f^{-1}(\{y\})}

\newcommand{\calH}{\mathcal{H}}
\newcommand{\calQ}{\mathcal{Q}}
\newcommand{\bbN}{\mathbb{N}}
\newcommand{\bbR}{\mathbb{R}}

\definecolor{grid_color}{rgb}{0.7529411764705882,0.7529411764705882,0.7529411764705882}
\definecolor{room_color}{rgb}{0.6,0.2,0}
\definecolor{real_obstacle_color}{rgb}{1,0.4980392156862745,0}
\definecolor{real_robot_color}{rgb}{0.4,0.3,0}
\definecolor{fake_obstacle_color}{rgb}{1,0.9,0.5}
\definecolor{fake_robot_color}{rgb}{0.8,0.6,0}

\title{A Note on the Time Complexity of Using Subdivision Methods for the Approximation of Fibers}
\author[$\dagger$]{Michael M. Bilevich}
\author[$\dagger$]{Dan Halperin}
\affil[$\dagger$]{Tel Aviv University}
\date{}

\begin{document}

\maketitle

\begin{abstract}
    Subdivision methods such as quadtrees, octrees, and higher-dimensional orthrees are standard practice in different domains of computer science. We can use these methods to represent given geometries, such as curves, meshes, or surfaces. This representation is achieved by splitting some bounding voxel recursively while further splitting only sub-voxels that intersect with the given geometry.
    It is fairly known that subdivision methods are more efficient than traversing a fine-grained voxel grid. In this short note, we propose another outlook on analyzing the construction time complexity of orthrees to represent implicitly defined geometries that are fibers (preimages) of some function. 
    This complexity is indeed asymptotically better  than traversing dense voxel grids, under certain conditions, which we specify in the note. In fact, the complexity is output sensitive, and is closely related to the Hausdorff measure and Hausdorff dimension of the resulting geometry. 
    
\end{abstract}

\section{Introduction}

Finding an explicit computer representation of an implicit manifold is a fundamental problem in computer algebra, computational geometry and computer graphics~\cite{botsch_polygon_2010}. 
Implicit manifolds arise naturally in different domains, like medicine~\cite{zhang_medical_2022}, computer-aided manufacturing and design (CAM/CAD)~\cite{samet_overview_1988} and robotics~\cite{bilevich_sensor_2023,lavalle_planning_2006}. Algebraic varieties, implicit manifolds defined as the zeros of polynomials, are also an interesting case intensively studied in computer algebra~\cite{burr_complexity_2017} and computational geometry~\cite{boissonnat_meshing_2006}. 
\smallskip

Given a function $f:\bbR^n \to \bbR^m$ and value $y \in \bbR^m$, an \emph{implicitly defined manifold} is defined as the \emph{fiber} (or \emph{preimage}) $\fiber = \{x : f(x) = y\}$. Although we use the term manifold, the preimage is not necessarily a manifold. A fundamental theorem of calculus, known as the ``implicit function theorem'' states that the fiber $\fiber$ is a manifold only if the value $y$ is regular, i.e., for all $x \in \fiber$, the differential $D_x f \neq 0$ is non-degenerate.\footnote{
In fact, the theorem also states that for any point $x\in\fiber$ with $D_x f \neq 0$, there is some small neighborhood $U_x \subset \bbR^n$ of $x$,
such that $\fiber \cap U_x$ is an $n-m$-manifold~\cite{krantz_implicit_2002}.}
In this work, we will assume the more general case, where the value $y$ may not be regular. Hence, we keep the general term \emph{fiber}.

\smallskip

We mention two common ways of achieving an explicit representation of fibers. One is by evaluating the implicit function on a dense voxel grid~\cite{bilevich_sensor_2023, botsch_polygon_2010,chen_neural_2021,lorensen_marching_1998, newman_survey_2006,nielson_dual_2004}, and the other is by subdividing some bounding voxel recursively, taking only sub-voxels that intersect with the geometry~\cite{an_measuring_2022,burr_complexity_2017,mobasseri_new_1992,samet_overview_1988}. Both methods are well-known and extensively used in various algorithms. In both ways, we strive for an approximation of the implicit manifold by voxels of diameter $\delta > 0$. Even though there are $O(\delta^{-n})$ such voxels in a dense voxel grid, only $O(\delta^{-d})$ voxels would intersect the implicit manifold, where $d$ is the \emph{Hausdorff dimension} of the manifold.

\section{Related Work}

\subsection{The Three Dimensional Case: Meshing of Implicit Surfaces}

One typical case is of implicitly defined \emph{surfaces}, which are the fibers\footnote{Fibers of some regular value in $\bbR$.} of smooth functions $f: \bbR^3 \to \bbR$. These surfaces, which are $2$-manifolds (with or without boundary), can be approximated as triangle soups or as polyhedral meshes~\cite{botsch_polygon_2010}. 
This problem is also known as ``meshing of implicit surfaces''. 
Provable methods based on Delaunay refinements exist for achieving topological correctness. Such techniques, however, require that the fiber be a closed $2$-manifold without boundary.

Other well-known and commonly used techniques are marching cubes, dual contouring, and their variants~\cite{chen_neural_2021,lorensen_marching_1998,newman_survey_2006,nielson_dual_2004}. Such techniques evaluate the implicit function $f$ on a dense voxel grid and use the results to construct triangular meshes that approximate the implicit surface.

\subsection{Quadtrees, Octrees, Orthrees and Soft Subdivision Search (SSS)}

Quadtrees, octrees, and orthrees\footnote{Orthrees are also known sometimes in literature as \emph{hyperoctrees}.} are tree data structures that represent a subdivision of two, three and higher-dimensional Euclidean spaces, respectively~\cite{campolattaro_quadtrees_2024,meagher_geometric_1982, samet_overview_1988}. 
Each node is a voxel, and children of that node are splits of that same voxel. Although the splitting predicate can be arbitrary, for the case of quadtrees and orthrees, it is common to split only voxels which intersect with the two or three-dimensional geometry.

The framework of soft subdivision search (SSS)~\cite{hsu_rods_2019,wang_soft_2013,yap_soft_2015} further presents techniques for solving various problems by similar subdivision, by using the notion of \emph{soft predicates}. This theoretical framework provides proofs and guarantees.

Already in the 1980s, it was known~\cite{meagher_geometric_1982, samet_overview_1988, DBLP:books/aw/Samet90} that for such tree data structures, we require significantly fewer voxels than in a dense voxel grid of the same resolution. For example, for two-dimensional quadtrees, there is the \emph{quadtree complexity theorem}~\cite{mobasseri_new_1992}, which states that the number of voxels covering a polygonal line is proportional to its length. A similar observation can be made for octrees and higher dimensional orthrees. A more recent work~\cite{burr_complexity_2017} also discusses the complexity of subdivision methods for fractals and relates it to the box-counting dimension (BCD), which in turn is related to the Hausdorff dimension~\cite{evans_measure_2018}. 

\section{The Subdivision Method}
\label{sec:subdiv-method}

In this section, we formally describe the specific subdivision method that we analyze. It is similar to orthrees and soft subdivision search. 
We show that this method converges to the desired result and in reasonable time,
 provided we can accurately and efficiently 
 evaluate whether a given voxel intersects the fiber we wish to approximate.
Of course, this is not always possible, which leads to the soft subdivision search method. In the current note, however, we assume that this predicate can be evaluated.

\bigskip

An axis-aligned voxel $V$ is defined as the Cartesian product of half-closed half-open intervals,
\begin{align}
    V = [a_1,b_1)\times \dots \times [a_n, b_n) \subset \bbR^n \;.
\end{align}
The diameter of a voxel is the largest Euclidean distance between any two points in the voxel. We denote the \emph{closure} of a voxel $V$ by $\Bar{V}$.

\bigskip

Let $f:\bbR^n \to \bbR^m$ be a function for which the following predicate is efficiently computable~\footnote{Any function for which we can evaluate the voxel intersection predicate (defined later) with some Turing machine.}: Determine whether a given axis-parallel cube intersects a fiber of the function. Let $y\in \bbR^m$ some value in the image of $f$. Denote by $M \coloneqq \fiber$ the fiber we approximate, and $\delta > 0$ is the desired voxel resolution. The output of our subdivision method is a collection of pairwise disjoint axis-aligned voxels $\{V_j\}_j$, such that $M \subset \bigcup_j \Bar{V}_j$, and the diameter of each voxel is at most $\delta$.
Let $V_0 \subset \bbR^n$ be an initial bounding volume of $M$, with diameter $\delta_0 = \mathrm{diam} V_0$.

The subdivision method is as follows: 

\bigskip

\begin{enumerate}
    \item Start with the collection $\calQ_0 = \{V_0\}$.
    \item For $t=1,\dots,N+1$, where $N = \lceil \log_2 (\delta_0 \delta^{-1}) \rceil$, do:
    \begin{enumerate}
        \item Set $\calQ_t = \{\}$.
        \item For each voxel $V_j \in \calQ_{t-1}$, test whether $\Bar{V_j} \cap M \neq \emptyset$.
        \item If the intersection is not empty, split the voxel $V_j$ into $2^n$ sub-voxels, which are obtained by splitting each dimension into two.
        Append their sub-voxels to the collection $\calQ_t$.
    \end{enumerate}
    \item Return the collection $\calQ_{N}$.
\end{enumerate}

By choice of $N$, we ensure that the diameter of each voxel in $\calQ_N$ is at most $\delta$. Furthermore, it is trivial to show that $M \subset \bigcup_{V_j \in \calQ_N} \Bar{V}_j$. Indeed, a stronger result can be shown: The set of approximations $\calQ_t$ converges to $M$ when taking $t \to \infty$.

\begin{theorem}
\label{thm:subdiv}
    Define the limit of $\calQ_t$ as:
    \begin{align}
        \lim_{t\to\infty} \calQ_t \coloneqq \lim_{t \to \infty} \bigcup_{V_j \in \calQ_t} \Bar{V}_j
        =
        \bigcap_{t=0}^\infty \bigcup_{V_j \in \calQ_t} \Bar{V}_j\;.
    \end{align}
    Then $M \subseteq \lim_{t \to \infty} \calQ_t$, and the difference $\left(\left(\lim_{t \to \infty} \calQ_t \right) \setminus M\right)$ is a set of measure zero.
\end{theorem}

\begin{proof}
    Presented in Appendix~\ref{sec:proof-convergence}.
\end{proof}

\section{Complexity Analysis}

We analyze the construction time of the representation of the fiber when using the subdivision method. We assume that querying whether a given voxel $V$ intersects the fiber $M$ takes $\alpha_M$ time, which depends only on the fiber $M$ and not on the resolution $\delta$.

Contrary to dense voxel grid approaches, the complexity of the subdivision method is \emph{output sensitive}, depending on the size of the output. More precisely, the complexity depends on the Hausdorff measure and Hausdorff dimension of the preimage.
In this section, we first discuss the notions of Hausdorff measure and dimension. Then, we formally define these notions and analyze the complexity of the subdivision method.

\smallskip

Recall~\cite{evans_measure_2018} the notion of the \emph{Lebesgue measure} of a set, which is intuitively defined as the infimum over the volume of a union of boxes that cover that set.
One can show that sub-manifolds of co-dimension $\geq 1$ are sets of Lebesgue measure zero. Nonetheless, some sets, such as the aforementioned manifolds, have a non-zero ``measure'', sometimes referred to as \emph{area} or \emph{length}. 
The notion of \emph{Hausdorff measure} extends this idea of measure to general sets, given some parameter $d$, called the \emph{dimension}.
The $d$-dimensional Hausdorff measure of a set is defined roughly as the infimum over the ``$d$-dimensional-volume'' of covers of that set.
The $d$-dimensional volume of a cube in $\bbR^n$ is (up to a constant) the diameter of the cube to the power of $d$.
For example, the $1$-dimensional Hausdorff measure of a curve embedded in $\bbR^n$ is its length, and the $2$-dimensional Hausdorff measure of a two-dimensional manifold embedded in $\bbR^3$ is its surface area.
The $n$-dimensional Hausdorff measure coincides, up to a constant, with the Lebesgue measure on $\bbR^n$. 
Finally, note that this dimension $d$ does not necessarily have to be integer. In fact, some sets can be shown to have a fractional Hausdorff dimension with a non-zero measure. In this work, however, we consider integral-valued dimensions.

\smallskip

We now formally define the notions of Hausdorff measure and Hausdorff dimension.

\begin{definition}
    Let $A \subset \bbR^n$, 
    let $\{B_j\}_j$ be any countable collection of subsets of $\bbR^n$, and let $d,\varepsilon$ be two real parameters:
    $0 \leq d < \infty$, $0 < \varepsilon \leq \infty$.
    Denote by $\mathrm{diam}(B)$ the diameter of some set $B$. Then define the $d$-dimensional-volume of the collection $\{B_j\}_j$ as
    \begin{align}
        \mu^d\left(\{B_j\}_j\right) \coloneqq 
            \sum_{j=1}^\infty (\mathrm{diam}B_j)^d \;,
    \end{align}
    and the $(d,\varepsilon)$-Hausdorff measure as the infimum of $d$-dimensional volumes covering the set $A$, where each subset $B_j$ is at most $\varepsilon$-wide:
    \begin{align}
        \calH_\varepsilon^d(A) \coloneqq 
            \inf \left\{
                \mu^d(\{B_j\}_j) \bigg\vert 
                A \subseteq \bigcup_{j=1}^\infty B_j, \ 
                \mathrm{diam}B_j \leq \varepsilon
            \right\} \;.
    \end{align}
    Taking $\varepsilon \to 0$, we get the Hausdorff $d$-dimensional measure:
    \begin{align}
        \calH^d (A) \coloneqq \lim_{\varepsilon\to 0} \calH_\varepsilon^d(A) = \sup_{\varepsilon > 0}  \calH^d_\varepsilon(A)\;.
    \end{align}
    The Hausdorff dimension of that set $A \subset \bbR^n$ is:
    \begin{align}
        \Hdim(A) \coloneqq
        \inf \{0 \leq s < \infty \big\vert \calH^s(A) = 0 \}\;.
    \end{align}
\end{definition}

We also recall~\cite{evans_measure_2018} the following fundamental result on Hausdorff dimension:

\begin{theorem}
    For a set $A \subset \bbR^n$, if $\Hdim = d$, the for any $0 \leq s < d$ the $s$-Hausdorff measure is zero, i.e., $\calH^s(A) = 0$. Furthermore, for any $s > d$, the $s$-Hausdorff measure is infinity, i.e., $\calH^s(A) = \infty$. Finally, the $d$-dimensional measure is a finite, non-zero value $0 < \calH^d(A) < \infty$.
\end{theorem}

\bigskip

We are now ready to analyze the complexity of the subdivision method. We first notice that there are at most $\lceil \log_2 \delta_0\delta^{-1} \rceil$ iterations, as we divide the voxel diameter by $2$ at each iteration.

Let $d := \Hdim(M)$ be the Hausdorff dimension of the fiber $M$. We then notice that by definition, $\mu^d(\calQ_t) = |\calQ_t| \cdot \delta_0^d \cdot 2^{-t\cdot d}$, where $|\calQ_t|$ is the number of voxels in $\calQ_t$. This is, in fact, the box-counting measure~\cite{burr_complexity_2017}. By Theorem~\ref{thm:subdiv}, $\lim_{t\to\infty} \calQ_t = M$, and this $M$ is a collection of manifolds, hence~\cite{evans_measure_2018}:
\begin{align}
    \calH^d(M) = \lim_{t\to\infty} \mu^d(\calQ_t) = 
    \lim_{t\to\infty} |\calQ_t| \cdot \delta_0^d \cdot 2^{-t\cdot d}\;.
\end{align}
We get that the number of voxels $|\calQ_t|$ at each iteration $t$ is $\Theta\left(\calH^d(M) \cdot  \delta_0^{-d}\cdot 2^{t \cdot d}\right)$.

\bigskip

Finally, during the construction of the approximation of $M$ by the subdivision method, we have $\Theta(\log \delta_0 \delta^{-1})$ iterations, and for each iteration, we evaluate the predicate (that takes $\alpha_M$ time) on $|\calQ_t|$ voxels. Hence, the construction time is (note that $\delta_0^{-d}\cdot 2^{td} \leq \delta_0^{-d} \cdot \delta_0 ^d \cdot \delta^{-d}$):
\begin{align}
    \Theta\left(
        \alpha_M \cdot \calH^d(M) \cdot \delta^{-d} \cdot \log (\delta_0\delta^{-1})
    \right) \;.
\end{align}

We note that the complexity of the output, which are the leaves of the orthree, is $\Theta\left(\calH^d(M) \cdot \delta^{-d}\right)$.

\bigskip\noindent
\begin{remark}
    Recall that we denote the diameter of the bounding box $V_0$ of $M$ by $\delta_0$. Then, to create a dense voxel grid of resolution $\delta$, we need to split each axis into at most $\lceil \delta_0 / \delta \rceil$ intervals.
    In other words, a $\delta$-resolution dense voxel grid of the bounding box $V_0$ has $\Theta\left( (\delta_0 / \delta)^ n \right)$ voxels.
    
    To construct a voxel grid approximation of the fiber $M$, we must iterate over all voxels in the dense grid and evaluate the voxel intersection predicate. Hence, the construction time of the representation based on the dense grid is
        $\Theta\left(\alpha_M \cdot \delta_0 ^ n \cdot \delta ^{-n}\right)$.

    This complexity is asymptotically worse than the subdivision method, when the Hausdorff dimension of the fiber $M$ is less than $n$.

\end{remark}

\section{Conclusion}

In this short note, we have reviewed a popular method for approximating implicitly defined fibers and analyzed its construction time. 
We see that the subdivision method is substantially more efficient than dense grid based methods, as it is both output sensitive, depending on the measure of the result, and is only exponential in the Hausdorff dimension of the result, instead of the dimension of the ambient space $\bbR^n$, i.e., $\delta^{-d}$ instead of $\delta^{-n}$. 

\bibliographystyle{plain}
\bibliography{bibliography}

\appendix

\section{Convergence of the Subdivision Method}
\label{sec:proof-convergence}

In this section we prove the convergence of the subdivision and dense voxel grid methods. That is, the approximation $\calQ_t$  of $M$ converge to $M$, up to some ``small'' set, i.e., a set of measure zero.

\begin{theorem}
    Define the limit of $\calQ_t$ as:
    \begin{align}
        \lim_{t\to\infty} \calQ_t \coloneqq \lim_{t \to \infty} \bigcup_{V_j \in \calQ_t} \Bar{V}_j
        =
        \bigcap_{t=0}^\infty \bigcup_{V_j \in \calQ_t} \Bar{V}_j\;.
    \end{align}
    Then $M \subseteq \lim_{t \to \infty} \calQ_t$, and the difference $\left(\left(\lim_{t \to \infty} \calQ_t \right) \setminus M\right)$ is a set of Lebesgue measure zero.
\end{theorem}

\bigskip

\begin{proof}
    To prove this theorem, we will show that indeed the fiber $M$ is contained in the limit. Then, we show that the limit is contained in the closure of $M$, concluding that the difference between the limit of the approximations $\calQ_t$ and $M$ is of 
    Lebesgue measure zero, as any set and its closure differ at a set of measure zero.

    \smallskip

    First, assume we have $p \in M$. Assuming $V_0$ is a bounding box of $M$, by induction, for every $t$, there is some voxel $V_{j_t} \in \calQ_t$ such that $p\in \Bar{V}_{j_t}$. Hence for all $t$, $p \in \bigcup_{V_j \in \calQ_t} \Bar{V}_j$, and $p \in \lim_{t\to \infty} \calQ_t$.

    \smallskip

    Now, assume we have $p \in \bigcap_{t=0}^\infty \bigcup_{V_j \in \calQ_t} \Bar{V}_j$. Then for all $t\in\bbN$, there is some voxel $V_t$ such that the intersection of its closure $\Bar{V}_t$ with $M$ is non empty. Hence there is a sequence of points on $M$, $p_1,p_2,\dots \in M$ such that $p_t \in \Bar{V}_t$. For each $t$, $||p_t - p|| \leq \mathrm{diam} V_t = \delta_0 2^{-t}$. Overall, we get that this sequence $p_t$ converges to the point $p$, and by definition of closure, $p \in \Bar{M}$ with $\Bar{M}$ is the closure of $M$.

    \smallskip
    
    Finally $M \subseteq \lim_{t\to\infty} \calQ_t \subseteq \Bar{M}$, and $\left(\lim_{t\to\infty} \calQ_t\right) \setminus M \subseteq \Bar{M} \setminus M$, hence it is of measure zero.
\end{proof}

\end{document}